\documentclass[11pt,reqno,a4paper]{article}
\usepackage{mathrsfs}
\usepackage{amsmath}
\usepackage{amsfonts}
\usepackage{amssymb}
\usepackage{latexsym}
\usepackage[dvips]{graphicx}
\usepackage[bookmarksnumbered=true, citecolor=blue, colorlinks=true, hypertex]{hyperref}

\newtheorem{theorem}{\bf Theorem}[section]
\newtheorem{lemma}[theorem]{\bf Lemma}
\newtheorem{proposition}[theorem]{\bf Proposition}

\newenvironment{proof}{\noindent{\em Proof:}}{\quad \hfill$\Box$\vspace{2ex}}

\newtheorem{definition}[theorem]{\bf Definition}

\def \bN {\Bbb N}
\def \bZ {\Bbb Z}

\def \bR {\Bbb R}

\def \cA {{\cal A}}
\def \cB {{\cal B}}

\def \cH {{\cal H}}

\def \and {\, \mbox{\rm and}\, }
\def \sinc {\,{\rm sinc}\,}

\def \supp {\,{\rm supp}\,}

\makeatletter

\newcommand{\Rmnum}[1]{\expandafter\@slowromancap\romannumeral #1@}
\makeatother
\begin{document}
\title{\bf Exponential Approximation of Multivariate Bandlimited Functions from Average Oversampling}
\author{Wenjian Chen\thanks{School of Mathematics
and Computational Science, Sun Yat-sen University, Guangzhou 510275,
P. R. China. E-mail address: {\it wenjianchen66@gmail.com}.}\quad and \quad
Haizhang Zhang\thanks{Corresponding author. School of Mathematics and Computational
Science and Guangdong Province Key Laboratory of Computational
Science, Sun Yat-sen University, Guangzhou 510275, P. R. China. E-mail address: {\it zhhaizh2@mail.sysu.edu.cn}. Supported in part by Natural Science Foundation of China under grants 11222103 and 11101438, and by the US Army Research
 Office.}}
\date{}
\maketitle
\begin{abstract}
Instead of sampling a function at a single point, average sampling takes the weighted sum of function values around the point. Such a sampling strategy is more practical and more stable. In this note, we present an explicit method with an exponentially-decaying approximation error to reconstruct a multivariate bandlimited function from its finite average oversampling data. The key problem in our analysis is how to extend a function so that its inverse Fourier transform decays at an optimal rate to zero at infinity.


\end{abstract}

\section{Introduction}
The main purpose of this note is to provide an explicit formula to reconstruct a multivariate bandlimited function from its finite average sampling data such that the associated approximation error decays exponentially to zero as the number of sample data increases. We begin with introducing the Paley-Wiener space of bandlimited functions and the average sampling strategy on it.

Let $d\in\bN$ be the dimension of the underlying Euclidean space and $\delta>0$ be the {\it bandwidth}. Denote by $\cB_\delta(\bR^d)$ the {\it Paley-Wiener} space of functions $f\in L^2(\bR^d)\cap C(\mathbb{R}^d)$ that are bandlimited to $[-\delta,\delta]^d$, namely, $\supp \hat{f}\subseteq [-\delta,\delta]^d$. Here $\hat{f}$ is the Fourier transform of $f$ that is defined as
$$
\hat{f}(\xi):=\frac{1}{(\sqrt{2\pi})^d}\int_{\mathbb{R}^d}f(x)e^{-i\langle x,\xi\rangle}dx, \ \ \ \xi\in\bR^d,
$$
where $\langle\cdot,\cdot\rangle$ is the standard inner-product on $\bR^d$. Each $\cB_\delta(\bR^d)$ is a Hilbert space after inheriting the norm of $L^2(\bR^d)$. Mathematical researches on the sampling theory originated from the celebrated Shannon sampling theorem \cite{C. E. Shannon,Whittaker}, which states that each $f\in\cB_\delta(\bR^d)$ can be exactly reconstructed from its function values sampled by the {\it Nyquist rate} $\frac\pi\delta$. Precisely, it holds for all $f\in\cB_\delta(\bR^d)$
\begin{equation}\label{shannonseries}
f(x)=\sum_{j\in\bZ^d}f(j\frac\pi\delta)\prod_{l=1}^d\frac{\sin(\delta x_l-\pi j_l)}{\pi(x_l-\frac\pi\delta j_l)}, \ \ \ x\in\bR^d,
\end{equation}
where the series converges absolutely and uniformly on $x\in\bR^d$. Many generalizations of the Shannon sampling theorem have been established (see, for example, \cite{GP,Hammerich1,HKK,MNS,NW,Zhang2009,Zhang2}).

We are concerned with the case when only finitely many sample data are available. Set $J_n:=[-n,n]^d\cap \bZ^d$ for $n\in\bN$. Looking at the Shannon series in (\ref{shannonseries}), let us assume that we have the finite sample data $\{f(j\frac\pi\delta):j\in J_n\}$ of some $f\in\cB_\delta(\bR^d)$. Naturally, one tends to truncate the Shannon series \eqref{shannonseries} as a manner of approximately reconstructing $f$. This turns out to be the optimal reconstruction method in the worst case scenario \cite{Rivlin,Charles A. Micchelli}. However, this method is of the slow approximation rate of $O(1/\sqrt{n})$, \cite{HT,D. Jagerman,Jordan,TI}. Dramatic improvement of the approximation rate can be achieved by using oversampling data. Here, oversampling means to sample at a rate strictly less than the Nyquist sampling rate $\pi/\delta$. Through a change of variables if necessary, we assume that the bandwidth $\delta<\pi$ and functions in $\cB_\delta(\bR^d)$ are sampled at the integer points, thus constituting oversampling as $1<\pi/\delta$.

It has been understood that one can reconstruct a {\it univariate} bandlimited function from its finite oversampling data with an exponentially decaying approximation error. Three such methods have been proposed in \cite{D. Jagerman,Charles A. Micchelli,Qian1}. The idea is to use a regularized Shannon series
$$
\sum_{j\in J_n}f(j)\sinc(x-n)\omega(x-n).
$$
to reconstruct $f\in \cB_\delta(\bR)$ from the finite oversampling data $\{f(j):-n\le j\le n\}$. Here, $\sinc t:=\sin(\pi t)/(\pi t)$. In \cite{D. Jagerman}, by letting $\omega(t)=\sinc^m((\pi-\delta)t/m)$ with $m=1+\lfloor n(\pi-\delta)/e\rfloor$, the approximation order of $O\left(\frac1n \exp(-\frac{\pi-\delta}en)\right)$ was obtained. Gaussian regularizers $\omega(t):=\exp(-\frac {t^2}{2r^2})$ were proposed in \cite{Wei2,Wei1}. The associated error analysis has been conducted in \cite{Qian1}. By letting $r=\sqrt{n/(\pi-\delta)}$, the approximation order of Gaussian regularized Shannon series was found to be $O\left(\sqrt{n} \exp(-\frac{\pi-\delta}2n)\right)$. In \cite{Charles A. Micchelli}, a spline function regularizer $\omega$ was used and the approximation order of $O\left(\frac1{\sqrt{n}}\exp(-\frac{\pi-\delta}2n)\right)$ was proved.

In practice, due to the limitation of the sampling machine, it is difficult to sample a function exactly at the integers. The following average sampling strategy
\begin{equation}\label{samplingstrategy}
\mu_j(f):=\int_{[-\frac\sigma2,\frac\sigma2]^d}f(t+j)d\nu(t), \ \ \ j\in \bZ^d,
\end{equation}
is more practical. Here, $\sigma>0$ and $\nu$ is a probability Borel measure on $[-\frac\sigma2,\frac\sigma2]^d$ that is usually discrete in real applications. Moreover, average sampling is more stable than sampling at a single point as the variance of the sampling noise tends to be reduced by the averaging process. For instance, highly robust reconstruction algorithms based on average sampling have been proposed in \cite{DaubechiesDevore}. There have been many extensions of the Shannon sampling theorem for average sampling \cite{Aldroubi2002,Aldroubi2005,Grochenig,SunZhou2002a,SunZhou2002}.

The major objective of this note is to present a method to reconstruct a multivariate function $f\in \cB(\bR^d)$ from its average oversampling data $\{\mu_j(f):j\in J_n\}$ such that the corresponding approximation error decays exponentially to zero as $n$ increases. We shall see that this question connects closely to the problem of smoothly extending a function so that the inverse Fourier transform of the extended function decays at an optimal rate at infinity. In the one-dimensional case, the problem is relatively easier to analyze as the region to be extended is only an interval. As result, an algorithm to exponentially reconstruct a univariate bandlimited function from its finite average oversampling data has recently been established in \cite{HZ}. In the multivariate case, the problem poses more difficulty as the boundary of the region to be extended is not just two points. The method in \cite{HZ} for the univariate case works only when the sampling probability $\nu$ in \eqref{samplingstrategy} is separated, namely, a tensor product of one-dimensional measures.

The rest sections are organized as follows. In Section 2, we present our approach and key problem. In Section 3, we provide a solution of the key problem for a general sampling probability $\nu$. As a result, we establish a method to exponentially reconstruct a multivariate bandlimited function $f\in\cB_\delta(\bR^d)$ from its average oversampling data $\{\mu_j(f):j\in J_n\}$. When $\nu$ is separated, the approximation error can be improved. We give analysis for this particular case in Section 4. Our main contribution is to explicitly construct a weight function $\Phi\in \cB_{2\pi-\delta}(\bR^d)$ such that
$$
\sup_{x\in(0,1)^d}\Big|f(x)-\sum_{j\in[-n,n]^d}\mu_j(f)\Phi(x-j)\Big|\le \|f\|_{L^2(\bR^d)}\frac{C}{\sqrt{n}}\exp\Big(-\frac n{e\rho}\Big),
$$
where $C$ and $\rho$ are two constants depending on $d,\delta,\sigma$. When $\nu$ is separated, $\rho$ depends on $\delta$ and $\sigma$ only. Detailed expressions of these two constants will be given in Theorems \ref{maintheorem} and \ref{main2}.

\section{The Approach and Key Problem}
We consider reconstructing a bandlimited function $f\in\cB_\delta(\bR^d)$ from its finite average sample data
\begin{equation}\label{samplingstrategy2}
\mu_j(f)=\int_{[-\frac\sigma2,\frac\sigma2]^d}f(t+j)d\nu(t), \ \ \ j\in J_n,
\end{equation}
where $\nu$ is a probability measure on $[-\sigma/2,\sigma/2]^d$. Our approach is to first have a complete reconstruction formula assuming that infinite sample data $\{\mu_j(f):j\in\bZ^d\}$ are available and later to truncate the formula to only use the finite data $\{\mu_j(f):j\in J_n\}$.

The inverse Fourier transform of the sampling measure $\nu$ is crucial in our analysis. Set
\begin{equation}\label{U}
U(\xi):=\int_{[-\frac\sigma2,\frac\sigma2]^d}e^{i\langle t,\xi\rangle}d\nu(t), \ \ \ \xi\in\bR^d.
\end{equation}
We shall assume that $\sigma$ is small enough so that $U$ is nonzero on $[-\delta,\delta]^d$. The precise restriction on $\sigma$ will be imposed later on.

We first seek a complete sampling reconstruction formula of the form
\begin{equation}\label{dualframeexpansion}
f=\frac{1}{(\sqrt{2\pi})^d}\sum_{j\in \bZ^d}\mu_{j}(f)\Phi(\cdot - j)
\end{equation}
for all $f\in\cB_\delta(\bR^d)$. This formulation is to be truncated. Thus, we should find a function $\Phi$ satisfying (\ref{dualframeexpansion}) and is fast-decaying at infinity. We first make two simple observations.

\begin{lemma}\label{mujframe}
It holds for all $f\in\cB_\delta(\bR^d)$
\begin{equation}\label{frameinequality}
\sum_{j\in\bZ^d}|\mu_j(f)|^2\le \|f\|_{L^2(\bR^d)}^2.
\end{equation}
\end{lemma}
\begin{proof}
It is well-known that $\cB_\pi(\bR^d)$ equipped with the norm of $L^2(\bR^d)$ is a reproducing kernel Hilbert space with the reproducing kernel
$$
\sinc(x):=\prod_{l=1}^d\frac{\sin\pi x_l}{\pi x_l},\ \ x\in\bR^d.
$$
In other words, it holds
$$
f(x)=\langle f,\sinc(x-\cdot)\rangle_{L^2(\bR^d)}\ \ \mbox{ for all }x\in\bR^d,\ f\in\cB_\pi(\bR^d).
$$
Moreover, $\sinc(j-\cdot)$, $j\in\bZ^d$ form an orthonormal basis for $\cB_\pi(\bR^d)$. Consequently, we have by the Parseval identity
\begin{equation}\label{mujframeboundeq1}
\|f\|_{L^2(\bR^d)}^2=\sum_{j\in\bZ^d}|\langle f,\sinc(j-\cdot)\rangle_{L^2(\bR^d)}|^2=\sum_{j\in\bZ^d}|f(j)|^2,\ f\in\cB_\pi(\bR^d).
\end{equation}
Now set $f\in\cB_\delta(\bR^d)$. Then $f\in\cB_\pi(\bR^d)$ as $\delta<\pi$. Note that $\cB_\pi(\bR^d)$ is translation-invariant. It implies that for each $t\in\bR^d$, $f(\cdot+t)\in\cB_\pi(\bR^d)$. We have by the definition (\ref{samplingstrategy2}) and the Cauchy-Schwartz inequality
$$
\sum_{j\in\bZ^d}|\mu_j(f)|^2= \sum_{j\in\bZ^d}\left|\int_{[-\frac\sigma2,\frac\sigma2]^d}f(t+j)d\nu(t)\right|^2\le \sum_{j\in\bZ^d}\int_{[-\frac\sigma2,\frac\sigma2]^d}|f(t+j)|^2d\nu(t).
$$
It follows from this inequality and (\ref{mujframeboundeq1})
$$
\sum_{j\in\bZ^d}|\mu_j(f)|^2\le \int_{[-\frac\sigma2,\frac\sigma2]^d}\sum_{j\in\bZ^d}|f(t+j)|^2d\nu(t)=\int_{[-\frac\sigma2,\frac\sigma2]^d}\|f(t+\cdot)\|_{L^2(\bR^d)}^2d\nu(t)=\|f\|_{L^2(\bR^d)},
$$
which completes the proof.
\end{proof}

To make sure that the series in (\ref{dualframeexpansion}) is well-defined, we shall choose $\Phi$ that is also bandlimited. To explain the reason, we need the notion of Bessel sequences.

\begin{definition}
Let $\cH$ be a separated Hilbert space. We call $\{f_j:j\in\bN\}\subseteq \cH$ a Bessel sequence in $\cH$ if there exists a positive constant $B$, called the Bessel bound for $\{f_j:j\in\bN\}$, such that for all $f\in\cH$
$$
\Big(\sum_{j=1}^\infty |\langle f,f_j\rangle_{\cH}|^2\Big)^{1/2}\le B\|f\|_{\cH}.
$$
\end{definition}

There is a useful characterization of Bessel sequences (see, \cite{Ole}, page 53).
\begin{lemma}\label{besselcharacter}
Let $\cH$ be a separated Hilbert space. Then $\{f_j:j\in\bN\}$ is a Bessel sequence in $\cH$ with Bessel bound $B$ if and only if for any $c=\{c_j:j\in\bN\}\in \ell^2$,
$$
\Big\|\sum_{j\in\mathbb{N}}c_jf_j\Big\|_{\cH}\leq B\|c\|_{\ell^2}.
$$
\end{lemma}

With the above preparation, we have the following observation.
\begin{lemma}\label{phibessel}
Let $\lambda>0$. It holds for all $x\in\bR^d$ and $\varphi\in\cB_\lambda(\bR^d)$
\begin{equation}\label{bessel}
\Big(\sum_{j\in\bZ^d}|\varphi(x-j)|^2\Big)^{1/2}\leq\lceil\frac\lambda \pi\rceil^{d/2}\|\varphi\|_{L^2(\bR^d)},
\end{equation}
where $\lceil\frac\lambda \pi\rceil$ is the smallest integer that is larger than or equal to $\frac\lambda \pi$.
\end{lemma}
\begin{proof}
We observe that for all $f\in\cB_\lambda(\bR^d)$ and $x\in\bR^d$
$$
f(x)=\langle f, K(x,\cdot)\rangle_{L^2(\bR^d)},
$$
where
$$
(K(x,\cdot))\hat{\,}(\xi):=\frac1{(\sqrt{2\pi})^d}e^{-i\langle x,\xi\rangle},\ \ \xi\in\bR^d.
$$
Thus,
$$
\sum_{j\in\bZ^d}|\varphi(x-j)|^2=\sum_{j\in\bZ^d}\Big|\langle \varphi(x-\cdot),K(j,\cdot)\rangle_{L^2(\bR^d)}\Big|^2.
$$
Therefore, (\ref{bessel}) can be confirmed by showing that $\{K(j,\cdot):j\in\bZ^d\}$ is a Bessel sequence in $\cB_\lambda(\bR^d)$ with Bessel bound $\lceil\frac\lambda \pi\rceil^{d/2}$. By Lemma \ref{besselcharacter}, it suffices to show that for all $c\in\ell^2(\bZ^d)$,
\begin{equation}\label{phibesseleq1}
\Big\|\sum_{j\in\bZ^d}c_jK(j,\cdot)\Big\|_{L^2(\bR^d)}^2\le \lceil\frac\lambda \pi\rceil^d \|c\|_{\ell^2}^2.
\end{equation}
To this end, we get by the Plancherel identity for the Fourier transform
$$
\Big\|\sum_{j\in\bZ^d}c_jK(j,\cdot)\Big\|_{L^2(\bR^d)}^2=\int_{[-\lambda,\lambda]^d}\Big|\sum_{j\in\bZ^d}c_j\frac1{(\sqrt{2\pi})^d}e^{-i\langle j,\xi\rangle}\Big|^2d\xi\le \lceil\frac\lambda \pi\rceil^d \int_{[-\pi,\pi]^d}\Big|\sum_{j\in\bZ^d}c_j\frac1{(\sqrt{2\pi})^d}e^{-i\langle j,\xi\rangle}\Big|^2d\xi.
$$
By the elementary fact that $\frac1{(\sqrt{2\pi})^d}e^{-i\langle j,\xi\rangle}$, $j\in\bZ^d$ form an orthonormal basis for $L^2([-\pi,\pi]^d)$,
$$
\int_{[-\pi,\pi]^d}\Big|\sum_{j\in\bZ^d}c_j\frac1{(\sqrt{2\pi})^d}e^{-i\langle j,\xi\rangle}\Big|^2d\xi=\|c\|_{\ell^2}^2.
$$
Combining the above two equations proves (\ref{phibesseleq1}) and completes the proof.
\end{proof}

We shall choose $\Phi\in \cB_{2\pi-\delta}(\bR^d)$. By Lemmas \ref{mujframe} and \ref{phibessel}, we get for all $f\in\cB_\delta(\bR^d)$ and $x\in\bR^d$
$$
\sum_{j\in\bZ^d}|\mu_j(f)\Phi(x-j)|\le \Big(\sum_{j\in\bZ^d}|\mu_j(f)|^2\Big)^{1/2}\Big(\sum_{j\in\bZ^d}|\Phi(x-j)|^2\Big)^{1/2}\le 2^{d/2}\|f\|_{L^2(\bR^d)}\|\Phi\|_{L^2(\bR^d)}.
$$
Therefore, the series in (\ref{dualframeexpansion}) converges absolutely. To ensure that it does equal $f$, we have the following necessary and sufficient condition.
\begin{lemma}\label{average1}
Let $\Phi\in \cB_{2\pi-\delta}(\bR^d)$. Then the identity (\ref{dualframeexpansion}) holds both pointwise and in $L^2(\bR^d)$ for all $f\in \cB_\delta(\bR^d)$ if and only if
\begin{equation}\label{ifandonlyif}
\hat{\Phi}(\xi)U(\xi)=1, \ \ \ \mbox{for almost every} \ \ \xi\in [-\delta, \delta]^d.
\end{equation}
\end{lemma}
\begin{proof}
Let $\Phi\in \cB_{2\pi-\delta}(\bR^d)$ and $f\in\cB_\delta(\bR^d)$. We see that the right hand side of \eqref{dualframeexpansion} converges in $L^2(\bR^d)$ to some $g\in \cB_{2\pi-\delta}(\bR^d)$ with the Fourier transform
$$
\hat{g}(\xi)=\hat{\Phi}(\xi)\frac{1}{(\sqrt{2\pi})^d}\sum_{j\in\bZ^d}\mu_j(f)e^{-i\langle j,\xi\rangle}=\hat{\Phi}(\xi)\int_{[-\frac\sigma2,\frac\sigma2]^d}\bigg(\frac{1}{(\sqrt{2\pi})^d}\sum_{j\in \bZ^d}f(t+j)e^{-i\langle j, \xi\rangle}\bigg)d\nu(t).
$$
Note that $\frac{1}{(\sqrt{2\pi})^d}\sum_{j\in \bZ^d}f(t+j)e^{-i\langle j, \xi\rangle}$ is the expansion of $\hat{f}(\cdot)e^{i\langle t,\cdot\rangle}$ with respect to the orthonormal basis $\big\{\frac{1}{(\sqrt{2\pi})^d}e^{-i\langle j, \xi\rangle}:j\in \bZ^d\big\}$ in $L^2([-\pi, \pi]^d)$. Consequently,
$$
\hat{g}(\xi)=\hat{\Phi}(\xi)(\hat{f}(\xi)U(\xi))_{2\pi}, \ \ \ \xi\in \bR^d,
$$
where the subindex $2\pi$ stands for the $2\pi$-periodic extension of a function originally defined only within $[-\pi, \pi]^d$. Thus, $\hat{g}$ equals $\hat{f}$ for all $f\in \mathcal{B}_{\delta}(\bR^d)$ if and only if \eqref{ifandonlyif} holds true. When (\ref{ifandonlyif}) is satisfied, as both sides in (\ref{dualframeexpansion}) are continuous functions on $\bR^d$, they also equal pointwise.
\end{proof}

Let $\Phi\in\cB_{2\pi-\delta}(\bR^d)$ satisfy (\ref{ifandonlyif}). Our method to reconstruct the values of a function $f\in\cB_\delta(\bR^d)$ on $(0,1)^d$ from its local finite average sample data $\{\mu_j(f): j\in J_n\}$ is directly given by
\begin{equation}\label{algorithm1}
(\cA_nf)(x):=\frac{1}{(\sqrt{2\pi})^d}\sum_{j\in J_n}\mu_j(f)\Phi(x-j), \ \ \ x\in(0, 1)^d.
\end{equation}

We have the following initial analysis of the approximation error for this reconstruction method.
\begin{proposition}\label{average16}
Let $\Phi\in\cB_{2\pi-\delta}(\bR^d)$ satisfy (\ref{ifandonlyif}). Then it holds for all $f\in\mathcal {B}_\delta(\bR^d)$ and $x\in(0,1)^d$,
\begin{equation}\label{abstractestimate}
|f(x)-(\mathcal {A}_nf)(x)|\leq\frac{1}{(\sqrt{2\pi})^d}\|f\|_{L^2(\bR^d)}\Big(\sum_{j\in\bZ^d\setminus J_n}|\Phi(x-j)|^2\Big)^{1/2}.
\end{equation}
\end{proposition}
\begin{proof}
Under the assumptions, (\ref{dualframeexpansion}) holds pointwise. Thus, for all $x\in (0,1)^d$,
$$
f(x)-(\cA_nf)(x)=\frac{1}{(\sqrt{2\pi})^d}\sum_{j\in\bZ^d\setminus J_n}\mu_{j}(f)\Phi(x-j).
$$
Applying the Cauchy-Schwartz inequality and the inequality \eqref{frameinequality} gives
$$
\begin{array}{ll}
|f(x)-(\cA_nf)(x)|&\displaystyle{\le \frac{1}{(\sqrt{2\pi})^d}\Big(\sum_{j\in \bZ^d\setminus J_n}|\mu_j(f)|^2\Big)^{1/2}\Big(\sum_{j\in\bZ^d\setminus J_n}|\Phi(x-j)|^2\Big)^{1/2}}\\
&\displaystyle{\leq\frac{1}{(\sqrt{2\pi})^d}\|f\|_{L^2(\bR^d)}\Big(\sum_{j\in\bZ^d\setminus J_n}|\Phi(x-j)|^2\Big)^{1/2}},
\end{array}
$$
as desired.
\end{proof}

By (\ref{abstractestimate}), to have an exponentially decaying approximation error, we should choose a $\Phi$ that decays really fast at infinity. We shall make use of a well-known relation between derivatives and the Fourier transform. For a multi-index $\alpha=(\alpha_1,\alpha_2,\cdots,\alpha_d)\in\bZ_+^d$, we set $|\alpha|:=\sum_{l=1}^d\alpha_l$ and denote by $D^\alpha$ the following differential operator
$$
D^\alpha=\frac{\partial^\alpha}{\partial x^\alpha}=\frac{\partial^{|\alpha|}}{\partial x_1^{\alpha_1}\cdots\partial x_d^{\alpha_d}}.
$$
For a multivariate polynomial
$$
P(x)=\sum_{\alpha}c_\alpha x^\alpha,
$$
we set
$$
P(D):=\sum_{\alpha}c_\alpha D^\alpha.
$$
Suppose that $\hat{\Phi}$ has sufficient regularity on $\bR^d$. Then it is well-known that
$$
\frac{1}{(\sqrt{2\pi})^d}\int_{[-2\pi+\delta,2\pi-\delta]^d}(P(D)\hat{\Phi})(\xi)e^{i\langle x-j,\xi\rangle}d\xi=P(i(j-x))\Phi(x-j).
$$
As a consequence,
\begin{equation}\label{usefulineq}
|\Phi(x-j)|\le\frac{1}{(\sqrt{2\pi})^d}\frac{\big\|P(D)\hat{\Phi}\big\|_{L^1([-2\pi+\delta,2\pi-\delta]^d)}}{|P(i(j-x))|}, \ \ \ j\in\bZ^d, \ x\in(0,1)^d.
\end{equation}
In conclusion, the key problem in our approach is to minimize for an appropriate differential operator $P(D)$ the quantity
\begin{equation}\label{L1norm}
\|P(D)\hat{\Phi}\big\|_{L^1([-2\pi+\delta,2\pi-\delta]^d)}
\end{equation}
subject to the complete reconstruction condition
\begin{equation}\label{ifandonlyif2}
\hat{\Phi}(\xi)=\frac1 {U(\xi)},\ \ \xi\in [-\delta,\delta]^d
\end{equation}
and that $\hat{\Phi}\in\cB_{2\pi-\delta}(\bR^d)$ has certain regularity on $\bR^d$. This minimization problem is hard to solve. When $d=1$, one only has to handle two conjunction points $-\delta$ and $\delta$ in extending $1/U$ smoothly from $[-\delta,\delta]$ to $[-2\pi+\delta,2\pi-\delta]$. In this case, \cite{HZ} gave an suboptimal solution by relaxing the $L^1$-norm in (\ref{L1norm}) to $L^2$-norm. An exponentially decaying approximation error was then obtained therein. In this note, we do not attempt to solve (\ref{L1norm}) either. Instead, we shall carefully extend $1/U$ to guarantee an exponentially decaying approximation error. When the measure $\nu$ is separated, the extension method in \cite{HZ} can be used via a tensor product form. This will be briefly discussed in Section 4. Our main concern is with a general sampling probability measure. The construction in \cite{HZ} does not work in this case. We present our extension method in the next section.

\section{General Sampling Probability Measures}
\setcounter{equation}{0}
Throughout this section, we let $d\ge 2$ and $\nu$ be a general sampling probability measure on $[-\frac\sigma2,\frac\sigma2]^d$. We assume that
\begin{equation}\label{widthcondition1}
(2\pi-\delta)\sigma d<\pi.
\end{equation}
Under this assumption, the crucial exponential function
$$
U(\xi)=\int_{[-\frac\sigma2,\frac\sigma2]^d}e^{i\langle t,\xi\rangle}d\nu(t)
$$
satisfies
\begin{equation}\label{gammacondition}
0<\gamma:=\cos\frac{(2\pi-\delta)\sigma d}{2}\le\int_{[-\frac\sigma2,\frac\sigma2]^d}\cos\langle\xi,x\rangle d\nu(x)\le|U(\xi)|\le1,\ \ \xi\in[-2\pi+\delta, 2\pi-\delta]^d.
\end{equation}

To extend $1/U$ from $I_\delta=[-\delta,\delta]^d$ to a smooth function on $\bR^d$ that is supported on $I_{2\pi-\delta}=[-2\pi+\delta,2\pi-\delta]^d$, our idea is to multiply $1/U$ by a smooth function that is identically equal to $1$ on $I_\delta$ and vanishes outside $I_{2\pi-\delta}$. The following choice will work for our purpose:
\begin{equation}\label{funV}
V(\xi):=\prod_{l=1}^dV_k(\xi_l), \ \ \ \xi=(\xi_1, \cdots,\xi_d)\in\bR^d,
\end{equation}
where
\begin{equation}\label{Vk}
V_k(s):=\left\{\begin{array}{cc}\displaystyle{d_k\int_{|s|}^{2\pi-\delta}\sin^{2k}\bigg(\frac{\pi(t-\delta)}
{2\pi-2\delta}\bigg)dt},&\quad\delta<|s|\le2\pi-\delta,\\
1,&\quad|s|\le\delta,\\
0, &\quad\mbox{otherwise,}\end{array}\right.
\end{equation}
and $d_k$ is chosen so that
\begin{equation}\label{dkconstant}
d_k\int_{\delta}^{2\pi-\delta}\sin^{2k}\bigg(\frac{\pi(t-\delta)}{2\pi-2\delta}\bigg)dt=1.
\end{equation}
Here, $k\in\bN$ represents the regularity order of $V$ that is to be optimally chosen. Our reconstruction function $\Phi$ is then determined by
\begin{equation}\label{Phik}
\hat{\Phi}(\xi)=\frac{V(\xi)}{U(\xi)}, \ \ \ \xi\in\bR^d.
\end{equation}
By our construction (\ref{Vk}) and (\ref{dkconstant}), $\Phi$ satisfies the complete reconstruction condition (\ref{ifandonlyif2}) and is $2k$-times continuously differentiable with respect to each of its variables.

We shall estimate the approximation error according to (\ref{abstractestimate}), where $|\Phi(x-j)|$ will be bounded by (\ref{usefulineq}). The differential operator $P(D)$ is set as
\begin{equation}\label{PD}
P(D):=(\frac{\partial^2}{\partial \xi_1^2}+\frac{\partial^2}{\partial \xi_2^2}+\cdots \frac{\partial^2}{\partial \xi_d^2})^k.
\end{equation}
Toward this purpose, we need to bound the $L^1$-norm of $P(D)\hat{\Phi}$. Several lemmas are needed. The first two of them will be used to bound the $L^\infty$-norm of the derivatives of each $V_l$, $1\le l\le d$.

\begin{lemma}\label{dkestimate}
It holds for all $k\in\bN$ and $\delta\in(0,\pi)$
\begin{equation}\label{dkestimateeq}
d_k\le\frac{\sqrt{2k}}{\pi-\delta}.
\end{equation}
\end{lemma}
\begin{proof}
A change of variables leads to
$$
\int_{\delta}^{2\pi-\delta}\sin^{2k}\left(\frac{\pi(t-\delta)}
{2\pi-2\delta}\right)dt=\frac{2\pi-2\delta}\pi\int_0^\pi\sin^{2k}tdt.
$$
Let
$$
b_{k}:=\int_0^\pi\sin^{2k}tdt
$$
and observe
$$
b_0=\pi,\ \  b_{k}=\frac{2k-1}{2k}b_{k-1}, \ \  k\geq 1.
$$
Hence, for $k\ge 1$,
$$
d_{k}=\frac1{2\pi-2\delta}\prod_{j=1}^k\frac{2j}{2j-1}.
$$
Since for $x>0$
$$
\ln\left(1+\frac1x\right)<\frac1x,
$$
we get that for $k\ge 1$
$$\ln\left(\prod_{j=1}^k\frac{2j}{2j-1}\right)=\displaystyle{\sum_{j=1}^k\ln\left(1+\frac1{2j-1}\right)}
\le\displaystyle{\ln2+\sum_{j=2}^k\frac1{2j-1}}\le
\displaystyle{\ln2+\frac12\int_1^{2k-1}\frac1xdx}=\displaystyle{\ln(2\sqrt{2k-1})}.
$$
The result of this lemma follows directly.
\end{proof}

\begin{lemma}\label{coro}
\label{sin}
It holds
\begin{equation}\label{coroinequality}
\left|\left(\sin^{2k}t\right)^{(j)}\right|\le 2^jk^j, \ \ t\in {\Bbb R},\ \ j\in\bZ_+.
\end{equation}
\end{lemma}
\begin{proof}
We recall the result due to Bernstein (see, \cite{GS}, page 5) that if $g$ is a trigonometric polynomial of the
form
$$
g(t)=a_0+a_1\cos t+b_1\sin t+\cdots+a_m\cos
mt+b_m\sin mt, \ \ t\in \bR
$$
satisfying
$$
\left|g(t)\right|\le1, \ \ t\in \bR
$$
then
$$
\left|g'(t)\right|\le m, \ \ t\in \bR.
$$
The inequality (\ref{coroinequality}) follows immediately from this celebrated fact.
\end{proof}

We next deal with the derivatives of $1/U$.
\begin{lemma}\label{Linfty1}
For all $\alpha\in\bZ_+^d$, it holds
\begin{equation}\label{linftynormU}
\Big\|D^\alpha (\frac1U)\Big\|_{L^\infty([-2\pi+\delta,2\pi-\delta]^d)}\le\frac{\sigma^{|\alpha|}}{2^{|\alpha|}\gamma^{{|\alpha|}+1}}d^{|\alpha|}2^{(d-1)|\alpha|}{\alpha}^{\alpha},
\end{equation}
where
$$
\alpha^\alpha:=\prod_{l=1}^d \alpha_l^{\alpha_l} \mbox{ with }0^0:=1.
$$
\end{lemma}
\begin{proof}
Observe
\begin{equation}\label{Uder}
\Big|(D^\alpha U)(\xi)\Big|=\Big|\int_{[-\frac\sigma2,\frac\sigma2]^d}(it_1)^{\alpha_1}\cdots(it_{d})^{\alpha_d}e^{i\langle\xi, t\rangle}d\nu(t)\Big|\le\frac{\sigma^{|\alpha|}}{2^{|\alpha|}},\ \ \xi\in\bR^d.
\end{equation}
Now we prove \eqref{linftynormU} by induction on $|\alpha|$. It is clearly true when $|\alpha|=0$. Suppose that it is true for all $0\le |\alpha|\le m-1$, $m\in\bN$. Now let $|\alpha|=m$. Put $h:=1/U$. Applying $D^\alpha$ to both sides of $1=hU$ and using the Leibniz formula, we have
$$
D^\alpha h=-\frac1U\sum_{\beta\in\bZ_+^d,\beta<\alpha}{\alpha \choose \beta}(D^{\alpha-\beta}U)(D^\beta h)
$$
where
$$
{\alpha \choose \beta}:=\prod_{l=1}^d{{\alpha_l} \choose {\beta_l}}
$$
and $\beta<\alpha$ means that $\beta_l\le \alpha_l$ for $1\le l\le d$ and $\beta\ne \alpha$. Therefore, by our induction, equations \eqref{gammacondition} and \eqref{Uder}, we get
\begin{equation}\label{linftynormUeq1}
|D^\alpha h|\le\frac1\gamma\sum_{\beta<\alpha}{\alpha \choose \beta}\frac{\sigma^{|\alpha|-|\beta|}}{2^{|\alpha|-|\beta|}}\frac{\sigma^{|\beta|}}{2^{|\beta|}\gamma^{{|\beta|}+1}}d^{|\beta|}2^{(d-1)|\beta|}{\beta}^{\beta}\le\
\frac{\sigma^{|\alpha|}d^{|\alpha|-1}}{2^{|\alpha|}\gamma^{{|\alpha|}+1}}2^{(d-1)(|\alpha|-1)}\sum_{\beta<\alpha}{\alpha \choose \beta}{\beta}^{\beta}.
\end{equation}
To continue, note that
$$
\{\beta\in \bZ_+^d:\beta<\alpha\}\subseteq \bigcup_{l=1}^d\{\beta\in\bZ_+^d:\beta_l<\alpha_l,\beta_m\le\alpha_m,m\ne l,1\le m\le d\}.
$$
As a consequence,
\begin{equation}\label{linftynormUeq2}
\sum_{\beta<\alpha}{\alpha \choose \beta}{\beta}^{\beta}\le \sum_{l=1}^d \Big[\sum_{\beta_l=0}^{\alpha_l-1}{{\alpha_l} \choose {\beta_l}}{\beta_l}^{\beta_l}\Big]\Big[\prod_{m\ne l}\sum_{\beta_m=0}^{\alpha_m}{{\alpha_m} \choose {\beta_m}}{\beta_m}^{\beta_m}\Big].
\end{equation}
We estimate
\begin{equation}\label{linftynormUeq3}
\sum_{\beta_l=0}^{\alpha_l-1}{{\alpha_l} \choose {\beta_l}}{\beta_l}^{\beta_l}\le \sum_{\beta_l=0}^{\alpha_l}{{\alpha_l} \choose {\beta_l}}{(\alpha_l-1)}^{\beta_l}=(1+(\alpha_l-1))^{\alpha_l}={\alpha_l}^{\alpha_l}
\end{equation}
and
\begin{equation}\label{linftynormUeq4}
\sum_{\beta_m=0}^{\alpha_m}{{\alpha_m} \choose {\beta_m}}{\beta_m}^{\beta_m}={\alpha_m}^{\alpha_m}+\sum_{\beta_m=0}^{\alpha_m-1}{{\alpha_m} \choose {\beta_m}}{\beta_m}^{\beta_m}\le 2{\alpha_m}^{\alpha_m}.
\end{equation}
Finally, we combine equations \eqref{linftynormUeq1}, \eqref{linftynormUeq2}, \eqref{linftynormUeq3}, and \eqref{linftynormUeq4} to obtain
$$
|D^\alpha h|\le \frac{\sigma^{|\alpha|}}{2^{|\alpha|}\gamma^{{|\alpha|}+1}}d^{|\alpha|-1}2^{(d-1)(|\alpha|-1)}d 2^{d-1}\alpha^{\alpha},
$$
which confirms (\ref{linftynormU}).
\end{proof}

We need one more preparation in order to bound the $L^1$-norm of $P(D)\hat{\Phi}$ on $[-2\pi+\delta,2\pi-\delta]^d$.

\begin{lemma}\label{Linfty2}
It holds for each $\alpha\in\bZ_+^d$
\begin{equation}\label{linftynormUV}
\Big\|D^\alpha \hat{\Phi}\Big\|_{L^\infty([-2\pi+\delta,2\pi-\delta]^d)}\le\frac{(1+\lambda)^d}\gamma\Big(
\frac{\sqrt2}{\pi\sqrt{k}}\Big)^d\prod_{l=1}^d\Big(\frac{\sigma d}{2\gamma}2^{d-1}\alpha_l+\frac{\pi k}{\pi-\delta}\Big)^{\alpha_l},
\end{equation}
where
\begin{equation}\label{constantbeta}
\lambda:=\frac{d\sigma 2^{d-5/2}(\pi-\delta)}{\gamma}.
\end{equation}
\end{lemma}
\begin{proof}
Firstly, we recall the definition (\ref{Vk}) and get by inequalities (\ref{dkestimateeq}) and (\ref{coroinequality})
$$
|V_k(t)|\le 1,\ \ |V_k^{(s)}(t)|\le  \frac{\sqrt{2}}{\pi\sqrt{k}}(2k)^s\Big(\frac{\pi}{2\pi-2\delta}\Big)^s\mbox{  for all }1\le s\le 2k,\ t\in [-2\pi+\delta,2\pi-\delta].
$$
It follows that $V(\xi)=\prod_{l=1}^dV_l(\xi_l)$ satisfies
\begin{equation}\label{linftynormUVeq1}
|(D^\beta V)(\xi)|\le \Big(\frac{\sqrt{2}}{\pi\sqrt{k}}\Big)^{\|\beta\|_0}\Big(\frac{\pi k}{\pi-\delta}\Big)^{|\beta|},\ \ \xi\in[-2\pi+\delta,2\pi-\delta]^d,\ \ \beta\in\bZ_+^d,
\end{equation}
where $\|\beta\|_0$ is the $\ell^0$-semi-norm of $\beta$, namely, the number of nonzero components of $\beta$. We then apply the Leibniz formula
$$
D^\alpha \hat{\Phi}=\sum_{\beta\in\bZ_+^d,\beta\le \alpha}{\alpha \choose\beta}\Big(D^\beta \frac1U\Big)D^{\alpha-\beta}V
$$
and estimates (\ref{linftynormU}), (\ref{linftynormUVeq1}) to get
\begin{equation}\label{linftynormUVeq2}
\begin{array}{ll}
\displaystyle{\|D^\alpha \hat{\Phi}\|_{L^\infty([-2\pi+\delta,2\pi-\delta]^d)}}&\displaystyle{\le \sum_{\beta\le \alpha}{\alpha\choose\beta}\frac{\sigma^{|\beta|}}{2^{|\beta|}\gamma^{{|\beta|}+1}}d^{|\beta|}2^{(d-1)|\beta|}{\beta}^{\beta}\Big(\frac{\sqrt{2}}{\pi\sqrt{k}}\Big)^{\|\alpha-\beta\|_0}\Big(\frac{\pi k}{\pi-\delta}\Big)^{|\alpha|-|\beta|}}\\
&=\displaystyle{\frac1\gamma\prod_{l=1}^d
\sum_{j=0}^{\alpha_l}{{\alpha_l}\choose j}
\frac{\sigma^{j}d^{j}}{2^{j}\gamma^{j}}2^{(d-1)j}{j}^{j}\Big(\frac{\sqrt2}{\pi\sqrt{k}}\Big)^{\min(1,\alpha_l-j)}\Big(\frac{\pi k}{\pi-\delta}\Big)^{\alpha_l-j}
}\\
&\le\displaystyle{\frac1\gamma\prod_{l=1}^d\Big(
\frac{\sigma^{\alpha_l}d^{\alpha_l}}{2^{\alpha_l}\gamma^{\alpha_l}}2^{(d-1)\alpha_l}{\alpha_l}^{\alpha_l}+\frac{\sqrt2}{\pi\sqrt{k}}\sum_{j=0}^{\alpha_l}{{\alpha_l}\choose j}
\frac{\sigma^{j}d^{j}}{2^{j}\gamma^{j}}2^{(d-1)j}{\alpha_l}^{j}\Big(\frac{\pi k}{\pi-\delta}\Big)^{\alpha_l-j}
\Big)}\\
&\le\displaystyle{\frac1\gamma\prod_{l=1}^d\Big(
\frac{\sigma^{\alpha_l}d^{\alpha_l}}{2^{\alpha_l}\gamma^{\alpha_l}}2^{(d-1)\alpha_l}{\alpha_l}^{\alpha_l}+
\frac{\sqrt2}{\pi\sqrt{k}}\Big(\frac{\sigma d}{2\gamma}2^{d-1}\alpha_l+\frac{\pi k}{\pi-\delta}\Big)^{\alpha_l}
\Big).}
\end{array}
\end{equation}
Note that
$$
\frac{\sqrt2}{\pi\sqrt{k}}\Big(\frac{\sigma d}{2\gamma}2^{d-1}\alpha_l+\frac{\pi k}{\pi-\delta}\Big)^{\alpha_l}\ge \frac{\sqrt2}{\pi\sqrt{k}}{{\alpha_l}\choose 1} \Big(\frac{\sigma d}{2\gamma}2^{d-1}\alpha_l\Big)^{\alpha_l-1}\frac{\pi k}{\pi-\delta}\ge\frac1\lambda\frac{\sigma^{\alpha_l}d^{\alpha_l}}{2^{\alpha_l}\gamma^{\alpha_l}}2^{(d-1)\alpha_l}{\alpha_l}^{\alpha_l}.
$$
It now follows from this and (\ref{linftynormUVeq2}) the desired inequality (\ref{linftynormUV}).
\end{proof}
\begin{lemma}
Let $P(D)$ be the differential operator given by (\ref{PD}) with $d\in\bN$. It holds
\begin{equation}\label{l1normUV}
\Big\|P(D) \hat{\Phi}\Big\|_{L^1([-2\pi+\delta,2\pi-\delta]^d)}\le\frac{(1+\lambda)^d}\gamma\Big(
\frac{\sqrt2}{\pi\sqrt{k}}\Big)^d(4\pi-2\delta)^dd^kk^{2k}\Big(\frac\pi{\pi-\delta}+2^{d-2}\frac{d\sigma}{\gamma}\Big)^{2k}.
\end{equation}
\end{lemma}
\begin{proof}
By the multinomial theorem,
$$
P(D)\hat{\Phi}=\sum_{\alpha\in\bZ_+^d,|\alpha|=k}\frac{k!}{\alpha_1!\cdots\alpha_d!}\frac{\partial^{2k}\hat{\Phi}}{\partial\xi_1^{2\alpha_1}\cdots\partial\xi_d^{2\alpha_d}}.
$$
By (\ref{linftynormUV}),
$$
\begin{array}{ll}
\displaystyle{\Big\|P(D) \hat{\Phi}\Big\|_{L^\infty([-2\pi+\delta,2\pi-\delta]^d)}}&\displaystyle{\le\frac{(1+\lambda)^d}\gamma\Big(
\frac{\sqrt2}{\pi\sqrt{k}}\Big)^d\sum_{|\alpha|=k}\frac{k!}{\alpha_1!\cdots\alpha_d!}\prod_{l=1}^d\Big(\frac{\sigma d}{2\gamma}2^{d-1}\alpha_l+\frac{\pi k}{\pi-\delta}\Big)^{2\alpha_l}}\\
&\displaystyle{\le\frac{(1+\lambda)^d}\gamma\Big(
\frac{\sqrt2}{\pi\sqrt{k}}\Big)^d\sum_{|\alpha|=k}\frac{k!}{\alpha_1!\cdots\alpha_d!}\prod_{l=1}^d\Big(\frac{\sigma d}{2\gamma}2^{d-1}k+\frac{\pi k}{\pi-\delta}\Big)^{2\alpha_l}}\\
&\displaystyle{=\frac{(1+\lambda)^d}\gamma\Big(
\frac{\sqrt2}{\pi\sqrt{k}}\Big)^dk^{2k}\Big(\frac\pi{\pi-\delta}+2^{d-2}\frac{d\sigma}{\gamma}\Big)^{2k}\sum_{|\alpha|=k}\frac{k!}{\alpha_1!\cdots\alpha_d!}}\\
&\displaystyle{=\frac{(1+\lambda)^d}\gamma\Big(
\frac{\sqrt2}{\pi\sqrt{k}}\Big)^dk^{2k}\Big(\frac\pi{\pi-\delta}+2^{d-2}\frac{d\sigma}{\gamma}\Big)^{2k}d^k}.
\end{array}
$$
The above equation together with
$$
\|P(D) \hat{\Phi}\Big\|_{L^1([-2\pi+\delta,2\pi-\delta]^d)}\le \|P(D) \hat{\Phi}\Big\|_{L^\infty([-2\pi+\delta,2\pi-\delta]^d)}(4\pi-2\delta)^d
$$
proves (\ref{l1normUV}).
\end{proof}

We are finally in a position to present our main result. Recall our assumption $(2\pi-\delta)\sigma d<\pi$ and constants
$$
\gamma=\cos\frac{(2\pi-\delta)\sigma d}{2}>0,\ \ \lambda=\frac{d\sigma 2^{d-5/2}(\pi-\delta)}{\gamma}.
$$
\begin{theorem}\label{maintheorem}
Let $\delta<\pi$, $\nu$ be a probability Borel measure on $[-\frac\sigma2,\frac\sigma2]^d$, $f\in\cB_\delta(\bR^d)$, and $\mu_j(f)$ be the average sampling data defined by (\ref{samplingstrategy2}). Construct weight function $\Phi$ by $\hat{\Phi}=V/U$, where functions $U$, $V$ are given by (\ref{U}) and (\ref{funV}), respectively. Suppose the number $n$ of sampling points satisfies
\begin{equation}\label{bign}
n\ge \frac83+e\rho\max(2,\frac {2d}3),\ \ \rho:=\frac{\sqrt{d}}{2}\Big(\frac\pi{\pi-\delta}+2^{d-2}\frac{d\sigma}{\gamma}\Big).
\end{equation}
Then with the parameter $k$ adaptively chosen according to $n$ as
\begin{equation}\label{optimalk}
k=\lceil \frac {n-2}{2e\rho}\rceil,
\end{equation}
the reconstruction method
$$
(\cA_n f)(x):=\sum_{j\in[-n,n]^d} \mu_j(f)\Phi(x-j),\ x\in(0,1)^d
$$
satisfies the approximation error
\begin{equation}\label{maintheoremeq}
\Big|f(x)-(\cA_n f)(x)\Big|\le\|f\|_{L^2(\bR^d)}\frac{C_{d,\delta,\sigma}}{\sqrt{n}}\exp\Big(-\frac n{e\rho}\Big),\ \ x\in(0,1)^d,
\end{equation}
where
\begin{equation}\label{constantC}
C_{d,\delta,\sigma}:=\sqrt2(e\rho)^{\frac{d+1}2}\exp(2+\frac2{e\rho})\frac{(1+\lambda)^d}\gamma\frac{(4\pi-2\delta)^d}{\pi^{2d}}\sqrt{2^d\omega_{d-1}}
\end{equation}
with $\omega_{d-1}$ denoting the area of the unit sphere in $\bR^d$.
\end{theorem}
\begin{proof}
Let $P(D)$ be given by (\ref{PD}). We apply the estimate (\ref{abstractestimate}) of the approximation error, equations (\ref{usefulineq}) and (\ref{l1normUV}) to get
\begin{equation}\label{maintheoremeq1}
\Big|f(x)-(\cA_n f)(x)\Big|\le \|f\|_{L^2(\bR^d)}\frac{(1+\lambda)^d}\gamma\frac{(4\pi-2\delta)^d}{2^{d/2}\pi^{2d}}\frac1{(\sqrt{k})^d}k^{2k}(2\rho)^{2k}\Big(\sum_{j\notin[-n,n]^d}\frac1{\|x-j\|_2^{4k}}\Big)^{1/2},
\end{equation}
where $\|\cdot\|_2$ denotes the standard Euclidean norm on $\bR^d$. Also use $\|\cdot\|_\infty$ to denote the $\ell^\infty$-norm on $\bR^d$. As $x\in(0,1)^d$,
$$
\sum_{j\notin[-n,n]^d}\frac1{\|x-j\|_2^{4k}}\le 2^d\int_{\tau\in\bR^d,\|\tau\|_\infty\ge n-1}\frac{d\tau}{\|\tau-x\|_2^{4k}}\le\int_{\|\tau\|_\infty\ge n-2}\frac{2^d}{\|\tau\|_2^{4k}}d\tau\le\int_{\|\tau\|_2\ge n-2}\frac{2^d}{\|\tau\|_2^{4k}}d\tau.
$$
Using polar coordinates, we have
\begin{equation}\label{maintheoremeq2}
\sum_{j\notin[-n,n]^d}\frac1{\|x-j\|_2^{4k}}\le\int_{\|\tau\|_2\ge n-2}\frac{2^d}{\|\tau\|_2^{4k}}dt=2^d\omega_{d-1}\int_{n-2}^\infty\frac{r^{d-1}}{r^{4k}}dr=2^d\omega_{d-1}\frac{(n-2)^{-4k+d}}{4k-d}.
\end{equation}
Combining (\ref{maintheoremeq1}) and (\ref{maintheoremeq2}) gives
$$
\Big|f(x)-(\cA_n f)(x)\Big|\le \|f\|_{L^2(\bR^d)} \frac{(1+\lambda)^d}\gamma\frac{(4\pi-2\delta)^d}{2^{d/2}\pi^{2d}}\frac{\sqrt{2^d\omega_{d-1}}}{\sqrt{4k-d}}\Big(\frac{n-2}k\Big)^{d/2}\varphi(k),
$$
where
$$
\varphi(t):=t^{2t}(2\rho)^{2t}(n-2)^{-2t},\ \ t>0.
$$
Elementary analysis implies that $\varphi$ attains its minimum on $(0,+\infty)$ at $t=(n-2)/(2e\rho)$. We hence choose $k$ as in (\ref{optimalk}) and obtain
$$
\Big|f(x)-(\cA_n f)(x)\Big|\le \|f\|_{L^2(\bR^d)} \frac{(1+\lambda)^d}\gamma\frac{(4\pi-2\delta)^d}{\pi^{2d}}\frac{\sqrt{2^de\rho\omega_{d-1}}}{\sqrt{2n-4-d{e\rho}}}(e\rho)^{d/2}\Big(\frac1e+\frac{2\rho}{n-2}\Big)^{\frac{n-2}{e\rho}}.
$$
Finally, noticing
$$
\Big(\frac1e+\frac{2\rho}{n-2}\Big)^{\frac{n-2}{e\rho}}=\exp\Big(-\frac{n-2}{e\rho}\Big)\Big[\Big(1+\frac{2e\rho}{n-2}\Big)^{\frac{n-2}{2e\rho}}\Big]^{2}\le e^2\exp\Big(-\frac{n-2}{e\rho}\Big)
$$
and that under (\ref{bign}),
$$
\sqrt{2n-4-d{e\rho}}=\sqrt{n}\Big(\frac {2n-4-de\rho}{n}\Big)^{1/2}\ge \frac{\sqrt{n}}{\sqrt{2}},
$$
we reach (\ref{maintheoremeq}) and complete the proof.
\end{proof}

When the sampling measure $\nu$ is a tensor product of one-dimensional measures, the exponential term in the above approximation error estimate can be improved. This is to be shown in the next section.

\section{Tensor Product Sampling Probability Measures}
\setcounter{equation}{0}
We consider a separated probability Borel measure $\nu$ in this section, which is the case in many applications. The average sampling data takes the following form
\begin{equation}\label{average0}
\mu_{j}(f):=\int_{-\frac\sigma2}^{\frac\sigma2}\cdots\int_{-\frac\sigma2}^{\frac\sigma2}f(t+j)d\nu_1(t_1)\cdots d\nu_d(t_d), \ \ \ j\in\bZ^d,
\end{equation}
where each $\nu_l$, $1\le l\le d$ is a probability Borel measure on $[-\sigma/2,\sigma/2]$. Again, we aim at fast reconstruction of $f\in\cB_\delta(\bR^d)$ from $\{\mu_j(f):j\in[-n,n]^d\}$ following the approach in Section 2.

The exponential function $U$ defined by (\ref{U}) in this case is of tensor product type. Namely,
\begin{equation}\label{tensorU}
U(\xi)=\prod_{l=1}^d U_l(\xi_l),\ \ \xi\in\bR^d,
\end{equation}
where
$$
U_l(t):=\int_{-\frac\sigma2}^{\frac\sigma2} e^{itx}d\nu_l(x), \ \ t\in\bR,\ \ 1\le l\le d.
$$
Our method of reconstructing $f(x)$, $x\in(0,1)^d$ is still via (\ref{algorithm1}), where $\Phi\in\cB_{2\pi-\delta}(\bR^d)$ satisfies (\ref{ifandonlyif}). Considering that $U$ is of tensor product type, we would like $\Phi$ to be a tensor product of univariate functions as well. This means to extend each $1/U_l$ from $[-\delta,\delta]$ to a function on $[-2\pi+\delta,2\pi-\delta]$ so that the inverse Fourier transform of the extended function decays fast at infinity. By the one-dimensional version of (\ref{usefulineq}), we need to well bound the $L^1$-norm of high order derivatives of the extended function. An extension method by replacing the $L^1$-norm in the key minimization problem corresponding to (\ref{L1norm}) with an $L^2$-norm was proposed in \cite{HZ}. In this paper, we shall use the approach in Section 3 instead. We will make comments on the differences between the two methods at the end of this section.

Therefore, the weight function $\Phi$ is determined by
\begin{equation}\label{weight2}
\hat{\Phi}(\xi):=\prod_{l=1}^d\hat{\phi_l}(\xi_l),\ \ \xi\in\bR^d,
\end{equation}
and
\begin{equation}\label{phil}
\hat{\phi_l}(\xi_l):=\frac{V_k(\xi_l)}{U_l(\xi_l)},\ \ \xi_l\in\bR,\ \ 1\le l\le d,
\end{equation}
where
$$
V_k(s):=\left\{\begin{array}{cc}\displaystyle{e_k\int_{|s|}^{2\pi-\delta}\sin^{k}\bigg(\frac{\pi(t-\delta)}
{2\pi-2\delta}\bigg)dt},&\quad\delta<|s|\le2\pi-\delta,\\
1,&\quad|s|\le\delta,\\
0, &\quad\mbox{otherwise}\end{array}\right.
$$
and $e_k$ is a constant such that
$$
e_k\int_{\delta}^{2\pi-\delta}\sin^{k}\bigg(\frac{\pi(t-\delta)}{2\pi-2\delta}\bigg)dt=1.
$$
Note that the $V_k$ is slightly different from that in \eqref{Vk}. It is $k$-times continuously differentiable here.

Thanks to the tensor-product form of $\Phi$ in (\ref{weight2}), the analysis will be of less difficulty than that in Section 3 and the resulting approximation error will be improved consequently.  Our assumption on the relation between the bandwidth $\delta$ and sampling width $\sigma$ is now relaxed to
\begin{equation}\label{widthcondition}
(2\pi-\delta)\sigma<\pi.
\end{equation}
And we also need a constant playing the important role similar to that of $\gamma$ in Section 3:
$$
\tilde{\gamma}:=\cos\frac{(2\pi-\delta)\sigma}2>0.
$$

Let us get started with bounding the $L^1$-norm of each $\hat{\phi_l}^{(k)}$.
\begin{lemma}\label{main2lemma}
Let $\phi_l$, $1\le l\le d$ be constructed by (\ref{phil}). Then it holds
\begin{equation}\label{main2lemmaeq}
\|\hat{\phi_l}^{(k)}\|_{L^1([-2\pi+\delta,2\pi-\delta])}\le \frac{1+\tilde{\lambda}}{\tilde{\gamma}}\frac{4(2\pi-\delta)}{\pi\sqrt{k}}k^k\Big(\frac{\tilde{\gamma}\pi+\sigma(\pi-\delta)}{2\tilde{\gamma}(\pi-\delta)}\Big)^k,
\end{equation}
where
$$
\tilde{\lambda}:=\frac14\sqrt{\frac{\sigma\pi(\pi-\delta)}{\tilde{\gamma}}}.
$$
\end{lemma}
\begin{proof}
Similar arguments as those in Lemmas \ref{dkestimate} and \ref{Linfty2} are able to yield
$$
\Big\|\Big(\frac1{U_l}\Big)^{(s)}\Big\|_{L^\infty([-2\pi+\delta,2\pi-\delta])}\le \frac1{\tilde{\gamma}}\frac{\sigma^s}{2^s\tilde{\gamma}^s}s^s,\ \ s\in\bZ_+,\ 1\le l\le d
$$
and
$$
|V_k^{(s)}(t)|\le \Big(\frac2{\pi\sqrt{k}}\Big)^{\min(s,1)}\Big(\frac{k\pi}{2\pi-2\delta}\Big)^s,\ \ 0\le s\le k,\ t\in\bR.
$$
The estimate (\ref{main2lemmaeq}) then follows from the above two equations and from
$$
|\hat{\phi_l}^{(k)}|\displaystyle{\le\sum_{s=0}^k{k\choose s}\Big|\Big(\frac1{U_l}\Big)^{(s)}\Big||V_k^{(k-s)}|}.
$$
Details are similar to those in the proof of Lemma \ref{Linfty2}.
\end{proof}

The approximation error for separated sampling measures is presented below.
\begin{theorem}\label{main2}
Let $\delta<\pi$, $\sigma>0$ satisfy (\ref{widthcondition}) and let $\Phi$ be constructed by (\ref{weight2}). Set
\begin{equation}\label{tilderho}
\tilde{\rho}:=\frac{\tilde{\gamma}\pi+\sigma(\pi-\delta)}{2\tilde{\gamma}(\pi-\delta)}\mbox{ and }\ k:=\lceil\frac{n-1}{e\tilde{\rho}}\rceil.
\end{equation}
The reconstruction method $\cA_n$ defined by (\ref{algorithm1}) has the approximation error
\begin{equation}\label{main2eq}
\Big| f(x)-(\cA_n f)(x)\Big|\le \|f\|_{L^2(\bR^d)}\frac{\tilde{C}_{d,\delta,\sigma}}{\sqrt{n}}\exp\Big(-\frac{n}{e\tilde{\rho}}\Big),\ x\in(0,1)^d,\ n\ge1+e\tilde{\rho},
\end{equation}
where
\begin{equation}\label{constanttildeC}
\tilde{C}_{d,\delta,\sigma}:=\sqrt{2e\tilde{\rho}}\exp(1+\frac1{e\tilde{\rho}})\frac{\sqrt{d}(1+\tilde{\lambda})(4\pi-2\delta)}{\tilde{\gamma}^{\frac{d+1}2}\pi^{\frac{d+3}{2}}}.
\end{equation}
\end{theorem}
\begin{proof}
By (\ref{abstractestimate}), we have for $x\in(0,1)^d$
\begin{equation}\label{main2eq1}
\Big| f(x)-(\cA_n f)(x)\Big|\le\frac{\|f\|_{L^2(\bR^d)}}{(\sqrt{2\pi})^d}\Big(\sum_{j\notin[-n,n]^d}\prod_{l=1}^d|\phi_l(x_l-j_l)|^2\Big)^{1/2}.
\end{equation}
Since
$$
\{j\in\bZ_+^d:j\notin[-n,n]^d\}\subseteq\bigcup_{l=1}^d\{j\in\bZ_+^d:|j_l|\ge n+1\},
$$
it holds
\begin{equation}\label{main2eq2}
\sum_{j\notin[-n,n]^d}\prod_{l=1}^d|\phi_l(x_l-j_l)|^2\le \sum_{l=1}^d \Big(\sum_{|j_l|\ge n+1}|\phi_l(x_l-j_l)|^2\Big)\prod_{m\ne l}\Big(\sum_{j_m\in\bZ}|\phi_m(x_m-j_m)|^2\Big).
\end{equation}
As each $\phi_m\in \cB_{2\pi-\delta}(\bR)$, by Lemma \ref{phibessel},
$$
\prod_{m\ne l}\Big(\sum_{j_m\in\bZ}|\phi_m(x_m-j_m)|^2\Big)\le\prod_{m\ne l} 2\|\phi_m\|_{L^2(\bR)}^2=2^{d-1}\prod_{m\ne l}\|\hat{\phi}_m\|_{L^2(\bR)}^2\le\frac{2^{d-1}}{\tilde{\gamma}^{d-1}}.
$$
The above equation together with (\ref{main2eq1}) and (\ref{main2eq2}) gives
\begin{equation}\label{main2eq3}
\Big| f(x)-(\cA_n f)(x)\Big|\le\frac{\|f\|_{L^2(\bR^d)}}{(\sqrt{2\pi})^d}\sqrt{\frac{2^{d-1}}{\tilde{\gamma}^{d-1}}}\Big( \sum_{l=1}^d \Big(\sum_{|j_l|\ge n+1}|\phi_l(x_l-j_l)|^2\Big)\Big)^{1/2},\ \ x\in(0,1)^d.
\end{equation}
Now use
$$
|\phi_l(x_l-j_l)|\le \frac1{\sqrt{2\pi}}\frac{\|\hat{\phi_l}^{(k)}\|_{L^1([-2\pi+\delta,2\pi-\delta])}}{|x_l-j_l|^k}
$$
and Lemma \ref{main2lemma} to obtain
\begin{equation}\label{main2eq4}
\Big( \sum_{l=1}^d \Big(\sum_{|j_l|\ge n+1}|\phi_l(x_l-j_l)|^2\Big)\Big)^{1/2}\le \frac1{\sqrt{2\pi}}\frac{1+\tilde{\lambda}}{\tilde{\gamma}}\frac{4(2\pi-\delta)}{\pi\sqrt{k}}k^k\tilde{\rho}^k
\sqrt{d}\Big(\sum_{|m|\ge n}\frac{1}{m^{2k}}\Big)^{1/2}.
\end{equation}
We estimate
$$
\sum_{|m|\ge n}\frac{1}{m^{2k}}\le 2\int_{n-1}^{+\infty}\frac1{t^{2k}}dt=\frac{(n-1)^{-2k+1}}{k-\frac12}.
$$
Combining the above equation with \eqref{main2eq3}, \eqref{main2eq4}, we reach
$$
\Big| f(x)-(\cA_n f)(x)\Big|\le \frac{\|f\|_{L^2(\bR^d)}}{(\sqrt{2\pi})^d}\sqrt{\frac{2^{d-1}}{\tilde{\gamma}^{d-1}}}\frac1{\sqrt{2\pi}}\frac{1+\tilde{\lambda}}{\tilde{\gamma}}\frac{4(2\pi-\delta)}{\pi\sqrt{k}}k^k\tilde{\rho}^k
\sqrt{d}\frac{(n-1)^{-k+\frac12}}{\sqrt{k-\frac12}}.
$$
The optimal choice $k=\lceil \frac{n-1}{e\tilde{\rho}}\rceil$ results in the desired estimate (\ref{main2eq}).
\end{proof}

We make a few comparisons between Theorems \ref{maintheorem} and \ref{main2} in order to explain the advantages brought by separated sampling measures. First, the important constant $\tilde{\rho}$ (\ref{tilderho}) in (\ref{main2eq}) is independent of the dimension $d$, while $\rho$ (\ref{bign}) in \eqref{maintheoremeq} will increase as $d$ increases, weakening the exponential approximation ability in (\ref{maintheoremeq}) for high dimensions. Besides, the constant $\tilde{c}_{d,\delta,\sigma}$ in (\ref{constanttildeC}) is much smaller than $c_{d,\delta,\sigma}$ in (\ref{constantC}). Finally, we compare the width condition $d\sigma(2\pi-\delta)<\pi$ for general sampling measures with $\sigma(2\pi-\delta)<\pi$ for separated sampling measures. Apparently, the first requirement implies that the sampling width needs to shrink as $d$ increases while the second one allows the width to keep the same for all dimensions.

Finally, we compare the construction in Section 4 when reduced to the one-dimensional case with that in \cite{HZ}. Each of the two methods has its pros and cons. The construction in \cite{HZ} requires solving a linear system with a Hilbert matrix as the coefficient matrix. Although the size of the Hilbert matrix is typically very small, the construction is somewhat inconvenient compared to (\ref{weight2}), which has a closed-form. In terms of the sampling width, \cite{HZ} requires $\sigma\delta<\pi$, which is better than $\sigma(2\pi-\delta)<\pi$ here.

\end{document}